\documentclass{birkjour}
 \newtheorem{thm}{Theorem}[section]
 \newtheorem{cor}[thm]{Corollary}
 \newtheorem{lem}[thm]{Lemma}
 \newtheorem{prop}[thm]{Proposition}
 \theoremstyle{definition}
 \newtheorem{defn}[thm]{Definition}
 \theoremstyle{remark}
 \newtheorem{rem}[thm]{Remark}
 \newtheorem*{ex}{Example}
 \numberwithin{equation}{section}
 
 \newtheorem{conj}[thm]{Conjecture}
 
 \usepackage{color}
 \usepackage{graphicx}
 
\usepackage{amsmath}
\usepackage{amssymb}
\usepackage{latexsym}

\usepackage{amsfonts,amsgen,amsopn}
\usepackage{epsfig}
\usepackage{mathpazo}

\newcommand{\T}{{\mathcal T}}
\newcommand{\sS}{{\mathcal S}}

\newcommand{\lL}{{\mathcal L}}

\newcommand{\pP}{{\mathcal P}}
\newcommand{\D}{{\mathcal D}}

\begin{document}

%
%
%
%
%
%
%
%
%

\title[Mathematical aspects of phylogenetic groves]
 {Mathematical aspects of phylogenetic groves}

\author[Mareike Fischer]{Mareike Fischer}

\address{
Center for Integrative Bioinformatics Vienna (CIBIV)\\
Max F. Perutz Laboratories\\ 
(Joint research institute of the \\
University of Vienna, Medical University of Vienna, Veterinary University of Vienna)\\
Dr. Bohr Gasse 9\\
A-1170 Vienna\\
Austria}

\email{email@mareikefischer.de}

\thanks{Financial support from the Wiener Wissenschafts-, Forschungs- und Technologiefonds (WWTF) to Arndt von Haeseler is greatly appreciated.}



\date{today}


\begin{abstract} The inference of new information on the relatedness of species by phylogenetic trees based on DNA data is one of the main challenges of modern biology. But despite all technological advances, DNA sequencing is still a time-consuming and costly process. Therefore, decision criteria would be desirable to decide a priori which data might contribute new information to the supertree which is not explicitly displayed by any input tree. A new concept, so-called groves, to identify taxon sets with the potential to construct such informative supertrees was suggested by An\'e et al. in 2009. But the important conjecture that maximal groves can easily be identified in a database remained unproved and was published on the Isaac Newton Institute's list of open phylogenetic problems. In this paper, we show that the conjecture does not generally hold, but also introduce a new concept, namely 2-overlap groves, which overcomes this problem.
\end{abstract}

\maketitle
\section{\sc Introduction}
\label{introduction}
One of the main challenges of biological sciences is the reconstruction of the `Tree of Life', i.e. the phylogenetic tree displaying all living species on earth. The genetic sequence data on some clusters of species are already available in databases like GenBank or SwissProt, and there are algorithms available to reconstruct the tree of each cluster. Unfortunately, many frequently used tree inference methods like maximum parsimony or maximum likelihood are known to be NP-hard (\cite{foulds_graham_1982}, \cite{roch_2006}, \cite{chor_tuller_2006}), which is why for such a huge number of species only heuristics can be used -- but these are usually not very reliable given such amounts of data. Therefore, the Tree of Life cannot be constructed all at once. Scientists rather depend on supertree methods to combine known phylogenies on fewer taxa (\cite{bininda_2004}, \cite{gordon_1986}, \cite{baum_1992}). But even if there is no conflict amongst the input trees, not all supertrees reveal new information: for instance, if the input trees have no shared taxa, they can be combined in any possible way and therefore do not lead to new conclusions on the relatedness of the species involved. In order to avoid this problem, An\'e et al. \cite{ane_eulenstein_2009} suggested to use the concept of groves -- sets of clusters with the potential to construct informative supertrees. As this potential merely depends on certain overlap properties of the input taxon sets, no a priori knowledge on underlying phylogenies is required. However, in order for groves to be useful for practical purposes, they should be easily identifiable in databases. Regarding this question, An\'e et al. state the following conjecture:

\begin{conj}  \label{conjecture} The following equivalent properties are true.
\begin{enumerate} 
\item For any set $\mathcal{S}$ of taxon sets, the set of maximal groves in $\mathcal{S}$ is a partiton of $\mathcal{S}$. 
\item If two groves intersect, their union is a grove.
\item Two maximal groves do not intersect.
\end{enumerate}
\end{conj}

Obiously, property (1) would reduce the search for maximal groves at least to a search on all possible partitions of the set of taxon sets under investigation, even if this might still be hard. Actually, the above conjecture raised a lot of attention. It was published both on the Isaac Newton Institute's list of open phylogenetic problems in 2007 (see http://www.newton.ac.uk/\\programmes/PLG/conj.pdf) as well as on the `Penny Ante' list of the Annual New Zealand Phylogenetics Meeting in Kaikoura in 2009 (see http://\\www.math.canterbury.ac.nz/bio/events/kaikoura09/penny.shtml). 

In this paper, we first show that the concept of groves can be simplified by introducing {\em tripartition groves}, and then we prove that unfortunately the conjecture is not in general true. We show this by presenting an explicit counterexample to property (2) of Conjecture \ref{conjecture}. We also prove that the conjecture even fails when making the definition of groves more restrictive to enforce informativeness. Despite these negative results, we also characterize a new concept of groves, namely {\em 2-overlap groves}, which guarantees property (2), and hence the whole conjecture, to hold. 

\section{Preliminaries} \label{preliminaries}
The main idea of the grove concept is to decide in advance, i.e. before even constructing any phylogeny, if a set of various taxon sets has the potential to deliver new information when being combined into one common supertree. In the context of rooted phylogenies, `new information' of a supertree refers to resolving at least one triple of taxa which is not resolved by any of the input trees, as triples are the smallest informative unit in the rooted setting. Such triples which get resolved by a supertree but not by any of the input phylogenies are called {\em resolved cross triples}. In order to define groves explicitly, we therefore need some formal definitions of a phylogeny and of (resolved) cross triples.

Recall that a {\it rooted binary phylogenetic $X$-tree} is a tree $\T =(V(\T),E(\T))$, with vertex set $V(\T)$ and edge set $E(\T)$, on a leaf (taxon) set $X=\{1,\ldots,n\} \subset V(\T)$ with only vertices of degree 1 (leaves) or 3 (internal vertices) and one vertex of degree 2, which is called {\it root}. In this paper, when there is no ambiguity we often just write `tree' or `phylogeny' when referring to a rooted binary phylogenetic $X$-tree. \\
A {\it topology assignment} $\pP$ on a set $\sS$ of taxon sets is a set of trees such that a) for each taxon set in $\sS$ there is exactly one tree in $\pP$ and b) all these trees are compatible, i.e. can be displayed by one common tree, which is then called {\it supertree}. 

We are now in a position to define cross triples. In the following, we denote by $\lL(\sS)$ the set of all taxa of a set of taxon sets $\sS$, i.e. $\lL(\sS)=\{x: \exists S \in \sS: x \in S\}$.

\begin{defn}[Cross triple] Let $\sS$ be a set of taxon sets $X_1,\ldots, X_m$ and let $\pi = S_1| \ldots | S_k$ be a partition of $\sS$, i.e. $\bigcup\limits_{i =1}^k  S_i = \sS$ and $S_i \cap S_j = \emptyset$ for all $i \neq j$. Then, a {\it cross triple of $\sS$ with respect to $\pi$} is a set of three taxa $\{ x,y,z\}\subset \lL(\sS)$ such that no $S_i$ contains all three of them. 
\end{defn}

Recall that for three taxa $x, y, z$ there are exactly three possible rooted binary phylogenetic tree topologies, namely $((x,y),z)$, $((x,z),y)$ and $((y,z),$ \\$ x)$ as shown in Figure \ref{triples}. A topology assignment on a set $\sS$ of taxon sets might lead to various possible supertrees, and these trees need not agree on how three taxa $x, y, z$ are related. In order to formally characterize this, we now define what it means if a cross triple is resolved.

 \begin{figure}[ht]      \centering\vspace{0.5cm} 
    \includegraphics[width=11cm]{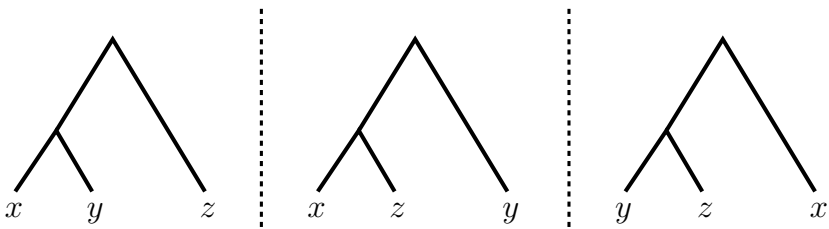} 
    \caption{Three taxa $x, y, z$ can be resolved in three different ways, which means three different rooted tree topologies are possible. }\label{triples}
  \end{figure}

\begin{defn}[Resolved cross triple] Let $\sS=X_1,\ldots, X_m$ be a set of taxon sets and let $\pi = S_1| \ldots | S_k$ be a partition of $\sS$. Let  $\{ x,y,z\} $ be a cross triple of $\sS$ with respect to $\pi$. $\{ x,y,z\} $ is called {\it resolved} if there is a topology assignment on $\sS$ such that all possible supertrees of this assignment display the same of the three possible rooted trees on $\{ x,y,z\} $.
\end{defn}

\begin{ex} \mbox{}\begin{enumerate} \item The set $\sS=\left\{ \left\{ v,w,x\right\}, \left\{ x,y,z\right\} \right\}$ does not have any resolved cross triple with respect to $\pi =\left\{ v,w,x\right\} | \left\{ x,y,z\right\} $, because no matter which topology assignment is chosen, the two input trees only intersect in only one taxon, namely $x$. Therefore, they can be combined in all possible ways such that no cross triple, like e.g. $\left\{v,w,y \right\}$, gets resolved (this is formally shown in \cite[Lemma 4.1]{ane_eulenstein_2009} as cited in Lemma \ref{ane}). \item For set $\tilde{\sS}=\left\{ \left\{ w,x,y \right\}, \left\{ x,y,z\right\} \right\}$ and partition $\tilde{\pi} =\left\{ w,x,y\right\} | \left\{ x,y,z\right\} $ there is a topology assignment $\pP$, namely $\pP=\{T_1,T_2\}$ as shown in Figure \ref{resolved}, such that all cross triples, like e.g. $\left\{w,y,z \right\}$, are resolved. In this case this is due to the fact that there is only one supertree. \end{enumerate}
\end{ex}

 \begin{figure}[ht]      \centering\vspace{0.5cm} 
    \includegraphics[width=11cm]{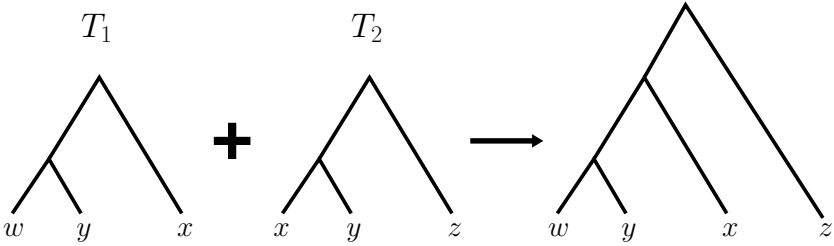} 
    \caption{For $\tilde{\sS}=\left\{ \left\{ w,x,y \right\}, \left\{ x,y,z\right\} \right\}$, the depicted topology assignment $\pP:=\{T_1,T_2\}$ leads to a unique supertree. Therefore, all cross triples with respect to the only possible partition of $\tilde{\sS}$, namely $\tilde{\pi }=\left\{ w,x,y\right\} | \left\{ x,y,z\right\}$, are resolved. }\label{resolved}
  \end{figure}

It is known that in the rooted case, either all possible supertrees of a topology assignment on a set $\sS$ of taxon sets agree on how a certain triple $\left\{x,y,z \right\} \subset \lL(\sS)$ should be resolved, or all three possible trees shown in Figure \ref{triples} are displayed in the set of supertrees \cite[Prop. 9.1]{bryant_steel_1995}. However, if the latter scenario happens for all possible cross triples, there is no information in the supertrees that is not already inherent in the topology assignment itself. We will now formalize this idea.
 
\begin{defn}[Informative topology assignment] Let $\sS$ be a set of taxon sets and $\pi$ a partition of $\sS$ such that there exists a cross triple of $\sS$ with respect to $\pi$. Then, a topology assignment $\mathcal{P}$ on $S$ is called {\it informative with respect to $\pi$} if some cross triple of $\sS$ with respect to $\pi$ is resolved by $\mathcal{P}$.
\end{defn}

\begin{ex}[continued] In the above example, for the set $\sS=\left\{ \left\{ v,w,x\right\}, \right. $ \\Ê$\left.\left\{ x,y,z\right\} \right\}$ there is no informative topology assignment. However, the set  $\tilde{\sS}=\left\{ \left\{ w,x,y \right\}, \left\{ x,y,z\right\} \right\}$ has an informative topology assignment with respect to partition $\tilde{\pi} =\left\{ w,x,y\right\} | \left\{ x,y,z\right\} $, namely the one shown in Figure \ref{resolved}. 
\end{ex}

As stated above, the idea now is to decide without considering specific input trees if a set of taxon sets has the potential to resolve cross triples. An\'e et al. showed in \cite{ane_eulenstein_2009} that this can be done with their grove concept, because the potential to resolve cross triples mainly depends on certain overlap properties of the underlying taxon sets. We now define groves as in \cite{ane_eulenstein_2009}.

\begin{defn}[Grove] \label{grove} A {\it grove} is a set $\sS$ of taxon sets such that for each possible partition $\pi$ of $\sS$ it holds either that
\begin{itemize}
\item there exists no cross triple of $\sS$ with respect to $\pi$, or
\item there exists an informative topology assignment on $\sS$ w.r.t. $\pi$. 
\end{itemize}
\end{defn}

\begin{ex}[continued] In the previous example, the set $\sS=\left\{ \left\{ v,w,x\right\}, \right. $ \\Ê$\left.\left\{ x,y,z\right\} \right\}$ is not a grove, because there is a partition of $\sS$, namely $\pi =\left\{ v,w,x\right\} | \left\{ x,y,z\right\} $, which has some cross triples but no resolved one for any topology assignment. On the other hand, the set  $\tilde{\sS}=\left\{ \left\{ w,x,y \right\}, \left\{ x,y,z\right\} \right\}$ is a grove because there is only one possible partition of  $\tilde{\sS}$, namely $\tilde{\pi} =\left\{ w,x,y\right\} | \left\{ x,y,z\right\} $, with respect to which there is an informative topology assignment as shown in Figure \ref{resolved}. 
\end{ex}

So groves have the property that if there is a cross triple with respect to a partition, at least one of these cross triples has to be resolved by some topology assignment. This implies that the supertree(s) on all taxon sets in a grove $\sS$ reveal some new information which is not present in any of the input phylogenies. But it has to be noted that the authors of \cite{ane_eulenstein_2009} by the first property in Definition \ref{grove} explicitly allow for the situation where there is no cross triple for any partition of $\sS$. This occurs, for instance, if $\sS$ contains one taxon set in which all taxa of $\sS$ are present. In this case, the supertree is always identical to the tree induced by this taxon set and therefore never reveals new information. The authors state that they include this case for biological reasons, which are not further explained. However, this means that while groves may have the potential to reveal new information, they are not guaranteed to do so. In Section \ref{informativegroves} we define informative and strictly informative groves in order to overcome this problem.

\par\vspace{0.3cm}

Next, in order to understand Conjecture \ref{conjecture}, we need to introduce a formal concept of maximality.

\begin{defn}[Maximal grove] \label{max} A grove $\sS$ in a database $\mathcal{D}$ of taxon sets  is called {\it maximal with respect to $\mathcal{D}$}, or {\it maximal} for short, if there is no taxon set $X$ in $\mathcal{D}$ such that $\sS \cup X$ is also a grove.
\end{defn}

\section{Results}
\subsection{Characterizing groves}

Our main goal in this section is to simplify the concept of groves by introducing tripartition groves. Moreover, we state some general properties of splits, partitions and groves. Our first two lemmas formalize the concept of cross triples and resolved cross triples for so-called splits. Recall that a split $\sigma=S_1|S_2$ of a set of taxon sets $\sS$ is a bipartition of $\sS$, i.e. a partition of $\sS$ into two disjoint parts.

\begin{lem} \label{crosstriples} Let $\sS$ be a set of taxon sets, let $\pi=S_1|\ldots|S_k$ be a partition of $\sS$ and let $\sigma=\hat{S}_1|\hat{S}_2$ be a split of $\sS$. \begin{enumerate} \item Assume there is a cross triple of $\sS$ with respect to $\pi$. Then, there are $i, j \in \{1,\ldots,k\}$, $i\neq j$, such that there are two taxa $s_i$, $s_j$ fulfilling the following properties: \begin{itemize} \item $s_i \in \lL(S_i)$, $s_i \notin \lL(S_j)$, and \item $s_j \in \lL(S_j)$, $s_j \notin \lL(S_i)$.\end{itemize}
\item Assume $|\lL(\sS)|\geq 3$ and assume there is no cross triple of $\sS$ with respect to $\sigma$. Then, either $\lL(\hat{S}_1)\subseteq \lL(\hat{S}_2)$ or $\lL(\hat{S}_2)\subseteq \lL(\hat{S}_1)$, i.e. either all taxa of $\hat{S}_1$ also lie in $\hat{S}_2$ or vice versa. \end{enumerate}
\end{lem}

\begin{proof} $\mbox{}$ \begin{enumerate} \item Let $\{x,y,z\}$ be a cross triple of $\sS$ with respect to $\pi$. By definition of cross triples, $x,y,z$ are not all contained together in any $S_m$, $m=1,\ldots,k$. Without loss of generality, choose $S_x \in \{S_1,\ldots,S_k\}$ such that $x\in \lL(S_x)$ and $y \notin \lL(S_x)$. However, as $y \in \sS$, there is an $S_y\in \{S_1,\ldots,S_k\}$ such that $y \in \lL(S_y)$. If possible, choose $S_y$ such that $x \notin \lL(S_y)$. Then, let $s_i:=x$, $s_j:=y$ and $S_i:=S_x$, $S_j:=S_y$. Else, if all sets which contain $y$ also contain $x$, choose a set $S_z$ such that $z \in \lL(S_z)$. Now if $y$ was contained in $S_z$, $S_z$ would also contain $x$, which contradicts the cross triple assumption. For the same reason, $z \notin \lL(S_y)$. So in this case, $y \in \lL(S_y)$, $y \notin \lL(S_z)$, $z \in \lL(S_z)$, $z \notin \lL(S_y)$. Then, let $s_i:=y$, $s_j:=z$ and $S_i:=S_y$, $S_j:=S_z$. This completes the proof. 
\item Assume there is no cross triple of $\sS$ with respect to $\sigma$. If not all taxa of $\hat{S}_1$ are contained in $\hat{S}_2$, this implies that there is a taxon $s_1$ such that $s_1 \in \hat{S}_1$ and $s_1 \notin \hat{S}_2$. If additionally not all taxa of $\hat{S}_2$ are contained in $\hat{S}_1$, this implies that there is a taxon $s_2$ such that $s_2 \in \hat{S}_2$ and $s_2 \notin \hat{S}_1$. Now as $|\lL(\sS)|\geq 3$, there is a $z \in \sS$ such that $z \neq s_1,s_2$. Then, by definition the set $\{s_1,s_2,z\}$ is a cross triple of $\sS$ with respect to $\sigma$. This contradicts the assumption and thus completes the proof. \qedhere
\end{enumerate}
\end{proof}

Next we recall a result by An\'e et al. in order to prove the following lemma.
\begin{lem}[Lemma 4.1 of \cite{ane_eulenstein_2009}] \label{ane} Let $X$ and $X'$ be two taxon sets that share at most one taxon. Then, any trees on $X$ and $X'$, respectively, are compatible and any cross triple of $\sS:= X\cup X'$ with respect to $\pi:= X|X'$ remains unresolved.
\end{lem}

\begin{lem} \label{splitlem} Let $\sS$ be a set of taxon sets and let $\sigma=S_1|S_2$ be a split of $\sS$. If there is a resolved cross triple of $\sS$ with respect to $\sigma$, then 
\begin{enumerate}\item \begin{enumerate} \item there is a taxon $s_1 \in S_1$ such that  $s_1 \notin S_2$ and \item there is a taxon $s_2 \in S_2$ such that  $s_2 \notin S_1$ and \item there are two taxa $x,y \in \lL(\sS)$, $x\neq y$, such that $x,y$ are both in $\lL(S_1)\cap \lL(S_2)$.\end{enumerate}
\item If properties 1. (a)-(c) hold and additionally there are taxon sets $X_1 \in S_1$ and $X_2 \in S_2$ such that $s_1, x, y \in X_1$ and $s_2, x, y \in X_2$, then there is a resolved cross triple of $\sS$ with respect to $\sigma$.\end{enumerate}
\end{lem}

\begin{proof} \begin{enumerate} \item By Lemma \ref{crosstriples}, all cross triples imply properties (1) and (2). Now we show that a resolved cross triple additionally implies (3). Assume there are no taxa $x,y$ with the properties specified in (3). This means that $S_1$ and $S_2$ share at most one taxon. Now we can choose any trees $T_1$ for $S_1$ and $T_2$ for $S_2$. As they overlap in at most one taxon, they are compatible by Lemma \ref{ane} and can be combined into a common supertree by pruning any edge of $T_1$ to any edge of $T_2$. This way, no cross triple can be resolved. So if a cross triple is resolved, (3) must hold.
\item
Now we show that (a), (b) and (c) together with the properties $s_1, x, y \in X_1 \subseteq S_1$ and $s_2, x, y \in X_2 \subseteq S_2$ imply that there is a resolved cross triple of $\sS$ with respect to $\sigma$. We prove this by construction of an informative topology assignment. We choose $T_1$ for $X_1$ and $T_2$ for $X_2$ such that they are so-called caterpillar trees as depicted by Figure \ref{caterpillar}.

 \begin{figure}[ht]      \centering\vspace{0.5cm} 
    \includegraphics[width=11cm]{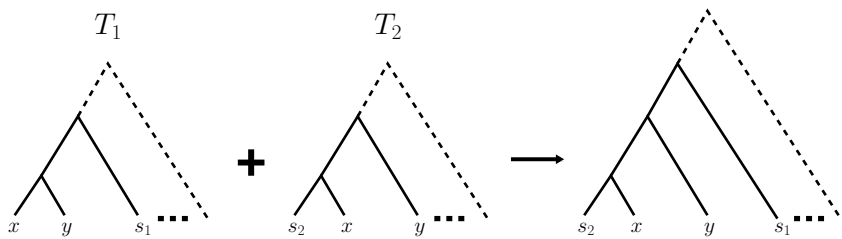} 
     \caption{Whenever the so-called caterpillar trees $T_1$ and $T_2$ are compatible and have the depicted taxon labellings, all supertrees combining these two trees will display the subtree $(((s_2,x), y),s_1)$ as shown on the right.  }\label{caterpillar}
  \end{figure}

In particular, we choose $T_1$ such that the taxa $x$ and $y$ specified by property (3) are together on the only 2-clade (so-called `cherry') and taxon $s_1$ (specified by properties (a) and (c)) is the taxon sharing the 3-clade with them. Moreover, we choose $T_2$ such that $s_2$ and $x$  (specified by properties (b) and (c)) are together on the 2-clade and taxon $y$ is the taxon sharing the 3-clade with them. The whole setting of $T_1$ and $T_2$ is shown in Figure \ref{caterpillar}. All other properties of $T_1$ and $T_2$ as well as all other trees for the other sets in $S_1$ and $S_2$ (if there are any) are chosen arbitrarily but in a way that leaves them compatible (i.e. if two sets share some taxa, we do not add conflicting information to the corresponding trees). By properties (a) and (b), the triple $\{s_1,s_2,x\}$, for instance, is a cross triple with respect to $\sigma$. Now every supertree comprising $T_1$ and $T_2$ (along with other trees if applicable) takes the fact into account that by $T_1$, $x$ and $y$ should be closer together than to any other taxon which is also included in $T_1$. Also, by $T_2$, $x$ and $s_2$ are closer to one another than any one of them is to $y$, but still closer to $y$ than to any other taxon of $T_2$. This implies the following subtree for each possible supertree: $(((s_2,x), y),s_1)$. Therefore, the cross triple $\{s_1,s_2,x\}$ is resolved as the only possible way to display this triple is $((s_2,x),s_1)$ as shown in Figure \ref{caterpillar}. Thus, we have constructed an informative topology assignment of $\sS$ with respect to $\sigma$. This completes the proof. \end{enumerate}\end{proof}

\noindent Next we introduce the definition of split groves and tripartition groves in order to simplify the definition of groves. 

\begin{defn}[Split grove] \label{splitsgrove} A {\it split grove} is a set $\sS$ of taxon sets such that for each possible split $\sigma$ of $\sS$ it holds either that
\begin{itemize}
\item there exists no cross triple of $\sS$ with respect to $\sigma$, or
\item there exists an informative topology assignment on $\sS$ w.r.t. $\sigma$. 
\end{itemize}
\end{defn}

\begin{rem} Obviously, every grove is also a split grove by definition. However, the opposite is not always true. This can be seen by looking at the set $\sS:= \{ \{x,y\},\{ y,z\}, \{ x,z\}\}$. None of the possible splits $\sigma_1:= \{x,y\}|\{ y,z\}, $ \\$ \{ x,z\}$, $\sigma_2:= \{x,y\},\{ y,z\}| \{ x,z\}$ or $\sigma_3:= \{ y,z\}|\{x,y\}, \{ x,z\}$ has any cross triples, as all taxa are displayed together on one side of the split, respectively. So by Definition \ref{splitsgrove}, $\sS$ is a split grove. However, the partition $\pi := \{x,y\}|\{ y,z\}| \{ x,z\}$ has a cross triple, namely $\{x,y,z\}$. This cross triple is not resolved by any tree combining the three 2-taxon sets (`cherries') as 2-taxon trees do not provide any information on the tree topology. Thus, $\sS$ is not a grove as there is a partition which has a cross triple, but no resolved one.
\end{rem}

As the above remark shows, the concept of split groves is not strong enough to cover the grove concept and to simplify it. Therefore, we next introduce tripartition groves and then show that they indeed are equivalent to groves. Recall that a tripartition is a partition of a set into three disjoint subsets.

\begin{defn}[Tripartition grove] \label{tripartitiongrove} A {\it tripartition grove} is a set $\sS$ of taxon sets such that for each possible split or tripartition $\tau$ of $\sS$ it holds either that
\begin{itemize}
\item there exists no cross triple of $\sS$ with respect to $\tau$, or
\item there exists an informative topology assignment on $\sS$ w.r.t. $\tau$. 
\end{itemize}
\end{defn} 
 
Again, all groves are tripartition groves by definition. Next we show that the converse is also true, which provides a significant simplification of the grove concept.

\begin{thm} \label{tripartitionthm} Every tripartition grove is a grove and every grove is a tripartition grove.
\end{thm}

\begin{proof} As mentioned before, the second part directly follows by definition. We now prove the first part. Therefore, let $\sS$ be a tripartition grove and let $\pi=S_1|\ldots|S_k$ be a partition of $\sS$. If there is no cross triple of $\sS$ with respect to $\pi$, there is nothing to show. So now assume that there is a cross triple $\{x,y,z\}$ of $\sS$ with respect to $\pi$. We have to show that there is also a resolved cross triple. Let $\sS_{xy}:= \bigcup \limits_{\stackrel{i \in \{1,\ldots,k\}}{x \in S_i, \mbox{\hspace{0.05cm}}z \notin S_i}} S_i$, \mbox{\hspace{0.05cm}}
 $\sS_{xz}:= \bigcup \limits_{\stackrel{i \in \{1,\ldots,k\}}{z \in S_i, \mbox{\hspace{0.05cm}}y \notin S_i}} S_i$, \hspace{0.5cm}
  $\sS_{yz}:= \bigcup \limits_{\stackrel{i \in \{1,\ldots,k\}}{y \in S_i, \mbox{\hspace{0.05cm}}x \notin S_i}} S_i$,\hspace{0.5cm}
   $\hat{\sS}:= \bigcup \limits_{\stackrel{i \in \{1,\ldots,k\}}{x,y,z \notin S_i}} S_i.$
 Now let $\tau:=  \sS_{xy}| \sS_{xz}| \sS_{yz} \cup \hat{\sS}  $ be a tripartition seperating $\{x,y,z\}$ such that it is a cross triple with respect to $\tau$ by construction. This is possible as $\{x,y,z\}$ is a cross triple of $\sS$ with respect to $\pi$ and thus the three taxa do not appear together in any $S_i$. Note that it is possible, if $\pi$ is a 2-partition, i.e. a split, that one of the sets $\sS_{xy},\sS_{xz}, \sS_{yz}$ as well as the set $\hat{\sS}$ may be empty.
However, as $\sS$ is a tripartition grove, the existence of a cross triple with respect to $\tau$ implies the existence of a resolved cross triple $\{\hat{x},\hat{y},\hat{z}\}$ of $\sS$ with respect to $\tau$. This cross triple is also a cross triple with respect to $\pi$ (otherwise $\{\hat{x},\hat{y},\hat{z}\}$ would appear together in one $S_i$ and would therefore not be separated by $\tau$, either). So using the same topology assignment that resolves $\{\hat{x},\hat{y},\hat{z}\}$ for $\tau$, $\{\hat{x},\hat{y},\hat{z}\}$ is a resolved cross triple of $\sS$ with respect to $\pi$. This completes the proof.
\end{proof}

Theorem \ref{tripartitionthm} shows that the concepts of groves and tripartition groves are identical. This characterization of groves thus simplifies the search for groves already drastically, as it reduces the analysis from all possible partitions to splits and tripartitions. 

In the following section, we finally investigate Conjecture \ref{conjecture}.

\vspace{1cm}

\subsection{Examining unions of intersecting groves}
Our main goal in this section is to show that Conjecture \ref{conjecture} does not generally hold. Therefore, we first show that indeed the three properties of the conjecture are equivalent, which enables us to provide a counterexample to property (2) and thereby disprove property (1).

\begin{lem} \label{equi} The three properties stated by Conjecture \ref{conjecture} are equivalent.
\end{lem}

\begin{proof} \mbox{}
 \begin{itemize}
\item We start by showing that (1) implies (2). Let $\sS_1, \sS_2 \subseteq \D$ be groves in a database of taxon sets $\D$ such that $\sS_1 \cap \sS_2 \neq \emptyset$. Assume $\sS_1 \cup \sS_2$ is not a grove. Let $\sS_1^{\mbox{\tiny max}}$,  $\sS_2^{\mbox{\tiny max}}$ be maximal supergroves of $\sS_1$, $\sS_2$, respectively. I.e. $\sS_1^{\mbox{\tiny max}}$,  $\sS_2^{\mbox{\tiny max}}$ are groves that contain $\sS_1$ or $\sS_2$, respectively, and by Definition \ref{max} they are such that no other taxon set of $\D$ can be added to them without destroying the grove property. Note that it is possible that $\sS_i = \sS_i^{\mbox{\tiny max}}$ for $i=1,2$. As by assumption $\sS_1 \cup \sS_2$ is not a grove, we have $\sS_i^{\mbox{\tiny max}} \neq \sS_1 \cup \sS_2$ for $i=1,2$.  Now as $\sS_1 \cap \sS_2 \neq \emptyset$, we have $\sS_1^{\mbox{\tiny max}} \cap \sS_2^{\mbox{\tiny max}}\neq \emptyset$. But because of property (1) of Conjecture \ref{conjecture}, the set of maximal groves of $\sS$ are a partition of $\sS$, which in particular means that they cannot intersect. So this is a contradiction and therefore the assumption is wrong. Thus, $\sS_1 \cup \sS_2$ is a grove.
\item Next we show that (2) implies (3). Let $\sS_1^{\mbox{\tiny max}}$,  $\sS_2^{\mbox{\tiny max}}$ be two maximal groves such that $\sS_1^{\mbox{\tiny max}} \neq \sS_2^{\mbox{\tiny max}}$. Assume $\sS_1^{\mbox{\tiny max}} \cap \sS_2^{\mbox{\tiny max}} \neq \emptyset$. 
Then by property (2), $\sS_1^{\mbox{\tiny max}} \cup \sS_2^{\mbox{\tiny max}}$ is a grove. But as $\sS_1^{\mbox{\tiny max}} \neq \sS_2^{\mbox{\tiny max}}$, there is at least one taxon set included in $\sS_1^{\mbox{\tiny max}} \cup \sS_2^{\mbox{\tiny max}}$ which is not present in $\sS_1^{\mbox{\tiny max}}$ or $\sS_2^{\mbox{\tiny max}}$, respectively. Therefore, $|\sS_1^{\mbox{\tiny max}} \cup \sS_2^{\mbox{\tiny max}}|> |\sS_1^{\mbox{\tiny max}}|$ and $|\sS_1^{\mbox{\tiny max}} \cup \sS_2^{\mbox{\tiny max}}|> |\sS_2^{\mbox{\tiny max}}|$. This contradicts the maximality of $\sS_1^{\mbox{\tiny max}}$, $\sS_2^{\mbox{\tiny max}}$. Thus, the assumption is wrong and $\sS_1^{\mbox{\tiny max}} \cap \sS_2^{\mbox{\tiny max}} = \emptyset$.
\item Last, we show that (3) implies (1). As (3) already states that maximal groves do not intersect, it only remains to show that each taxon set $X \in \D$ belongs to a maximal grove. All $X \in \D$ are groves by definition (as for single sets there is no partition). Now each such grove can either be combined with other taxon sets in $\D$ in order to form a bigger grove (which then can be maximized adding more taxon sets if possible), or no such combination is possible. In the latter case, $X$ itself is maximal by definition. So in both cases, $X$ belongs to a maximal grove. \qedhere
 \end{itemize}  \end{proof}
 
We are now in a position to show that Conjecture \ref{conjecture} does not hold. 

\begin{prop} \label{counterex} Conjecture \ref{conjecture} is not generally true. In particular, the union of intersecting groves is not necessarily a grove.
\end{prop}
 
\begin{proof} We provide an explicit counterexample to property (2) of Conjecture \ref{conjecture}. Let $\sS_1:= \left\{ \left\{ 1,2,3Ê\right\}, \left\{ 1\right\} \right\}$ and $\sS_2:= \left\{ \left\{ 1,4,5Ê\right\}, \left\{ 1\right\} \right\}$. Then, $\sS_1$ is a grove as there is only one possible partition, namely $\pi := \left\{ 1,2,3Ê\right\} | \left\{ 1\right\}$, and for this partition there are no cross triples. By the same argument, $\sS_2$ is also a grove. Now $\sS_1 \cap \sS_2 =\left\{ \left\{ 1\right\}\right\} \neq \emptyset$. If Conjecture \ref{conjecture} was true, $\sS_1 \cup \sS_2 = \left\{ \left\{ 1,2,3Ê\right\}, \left\{ 1\right\},\right.$\\$\left.\left\{ 1,4,5Ê\right\} \right\}$ would be a grove. But by Lemma \ref{ane}, this cannot be the case, as the sets $\{1,2,3\}$ and $\{1,4,5\}$ only share one taxon and thus all cross triples between the two sets remain unresolved. Moreover, the set $\{1\}$ does not contribute any new information to either of the trees. Therefore, $\sS_1 \cup \sS_2$ is not a grove.
\end{proof}

In our proof, we explicitly use Lemma 4.1 of \cite{ane_eulenstein_2009} in order to show that the example provided is a counterexample to property (2) of Conjecture \ref{conjecture}, but of course one can also examine all possible topology assignments for the sets $\left\{ 1,2,3Ê\right\}$ and $\left\{ 1,4,5Ê\right\}$ (the singleton $\left\{ 1Ê\right\}$ does not contribute any information). The only possible cross triples of $\sS_1 \cup \sS_2$ (if there are any, which is only the case for partitions which split the sets $\left\{ 1,2,3Ê\right\}$ and $\left\{ 1,4,5Ê\right\}$ apart from each other) are $\{1,2,4 \}$,  $\{1,2,5 \}$, $\{1,3,5 \}$, $\{2,3,4 \}$, $\{2,3,5 \}$, $\{2,4,5 \}$ and $\{3,4,5 \}$. Now the intuitive reason why none of these cross triples is resolved when combining the topologies assigned to the sets $\left\{ 1,2,3Ê\right\}$ and $\left\{ 1,4,5Ê\right\}$ is that their only overlap is taxon $1$, which does not suffice to fix anything in the possible supertrees. 

\subsection{Some corrections concerning the paper by An\'e et al. \cite{ane_eulenstein_2009} }

It should be noted that $\sS_1:= \left\{ \left\{ 1,2,3Ê\right\}, \left\{ 1\right\} \right\}$ and $\sS_2:= \left\{ \left\{ 1,4,5Ê\right\}, \left\{ 1\right\} \right\}$ as presented in the proof of Proposition \ref{counterex} are indeed both groves, despite the fact that each of them consists only of two taxon sets which overlap with one another in just one taxon. This is in contrast to a statement in the context preceding Proposition 5.1 of \cite{ane_eulenstein_2009}, where the authors claim that the smallest example of a grove with one-taxon overlaps consists of four taxon sets. In fact, the authors make the same mistake in both Propositions 5.1 and 5.2. Therefore, we will now formally state these Propositions and correct them thereafter.

\begin{bfseries}Proposition 5.1 of \cite{ane_eulenstein_2009} \end {bfseries} Let $\sS$ be a set of three taxon sets such that no two elements of $\sS$ share more than one taxon.
\begin{enumerate}\item Then, no cross triple of $\sS$ with respect to any partition is resolved, and
\item therefore, $\sS$ is not a grove.
\end{enumerate}

The second statement in Proposition 5.1 of \cite{ane_eulenstein_2009} is wrong. In fact, the authors correctly prove the first statement but then conclude the second statement out of the first without taking into account that their definition of groves does not require the existence of any cross triples. The first part of the proposition, however, is correct and interesting; it is another formal statement proving that $\sS_1 \cup \sS_2$ as in the proof of Proposition \ref{counterex} cannot be a grove. This is due to the fact that this set can be partitioned such that it has cross triples, but as it consists of three taxon sets where no element shares more than one taxon with another element, no such cross triple can be resolved. 

Nonetheless, there are groves consisting of three taxon sets where no two elements share more than one taxon. For instance, $\sS:=\{\{1,2,3\}, \{1,2\}, $ \\$ \{2,3\}\}$ is a grove contradicting statement (2) of Proposition 5.1 of \cite{ane_eulenstein_2009}. This is due to the fact that one set in $\sS$ contains all taxa, so no partition contains any cross triples. 

As the authors use Proposition 5.1 of \cite{ane_eulenstein_2009} to prove the following proposition, it is not surprising that the same mistake occurs there, too. \par\vspace{0.5cm}
 
\begin{bfseries}Proposition 5.2 of \cite{ane_eulenstein_2009}\end{bfseries}Let $\sS$ be a set of four taxon sets such that no two elements of $\sS$ share more than one taxon. Then, $\sS$ is a grove if and only if 
\begin{enumerate}\item each taxon set shares a taxon with each other taxon set, and
\item each of the six overlaps involves a different taxon.
\end{enumerate}

The statement of Proposition 5.2 of \cite{ane_eulenstein_2009} can be proven wrong by considering the set $\sS:=\{\{1,2,3\}, \{1,2\},\{1,3\},\{1\}\}$. Again, as one taxon set in $\sS$ contains all taxa, no partition produces any cross triples, which makes $\sS$ a grove. But all six overlaps involve the same taxon, namely taxon $1$. However, the other direction of Proposition 5.2 of \cite{ane_eulenstein_2009} is true, as well as both directions are true for groves which do indeed have at least one cross triple with respect to a partition.

The mistake made in the above propositions is explicable considering the main purpose of the grove idea: only sets of taxon sets with the potential to deliver new information should be considered. As explained before, this idea is not really reflected in Definition \ref{grove}, the definition of groves as given by the authors. However, it will be shown in Section \ref{informativegroves} that simply excluding trivial cases (like e.g. groves containing a set of all taxa) from the definition does unfortunately not improve the situation with regards to Conjecture \ref{conjecture}.

\subsection{Informative groves and strictly informative groves} \label{informativegroves}
\par \vspace{0.5cm}
Now we modify the concept of groves in order to guarantee more informativeness. We have seen for instance that the example in the proof of Proposition \ref{counterex} leading to the problematic case only employs groves $\sS_1$ and $\sS_2$ which both contain one set which includes all their taxa, respectively. So any supertree of the sets in $\sS_1$ will always be identical to the tree assigned to this taxon set (and the same holds for $\sS_2$, respectively). Therefore, whilst both $\sS_1$ and $\sS_2$ fulfill the requirements of groves, they do not represent very interesting cases as no supertree of them can ever reveal new information. We now examine Conjecture \ref{conjecture} for more interesting cases, namely informative and strictly informative groves.

\begin{defn}$\mbox{}$
\begin{enumerate} 
\item A grove $\sS$ is called {\it informative}, if there exists a partition $\pi$ of $\sS$ such that there is an informative topology assignment on $\sS$ w.r.t. $\pi$.   \item A grove $\sS$ is called {\it strictly informative}, if for {\bf all} partitions $\pi$ of $\sS$ there is an informative topology assignment on $\sS$ w.r.t. $\pi$. \end{enumerate}
\end{defn}

We now state a result analogous to Proposition \ref{counterex} for informative groves.

\begin{prop} \label{counterex2} The union of intersecting informative groves is not necessarily a grove.
\end{prop}

\begin{proof} We provide an explicit counterexample. Let $\sS_1:= \left\{ \left\{ 1,2,3Ê\right\}, \left\{ 2,3,4Ê\right\}, \right. $ \\ $ \left. \left\{ 4\right\} \right\}$ and $\sS_2:= \left\{ \left\{ 4,5,6Ê\right\}, \left\{ 5,6,7Ê\right\}, \left\{ 4\right\} \right\}$. Then, $\sS_1$ is an informative grove as the only way to partition $\sS_1$ such that there are cross triples is to split $\left\{ 1,2,3Ê\right\}$ and  $\left\{ 2,3,4Ê\right\}$ apart. Then, the cross triples are $\left\{ 1,2,4Ê\right\}$ and $\left\{ 1,3,4Ê\right\}$, which are for instance resolved when the topology assignment $((1,2),3)$ and $((2,3),4)$ is chosen, as then the only possible supertree is $(((1,2),3),4)$. This is analogous to the example depicted by Figure \ref{resolved}. By the same argument, $\sS_2$ is an informative grove. Now $\sS_1 \cap \sS_2 =\left\{ \left\{ 4\right\}\right\} \neq \emptyset$. But $\sS_1 \cup \sS_2 = \left\{ \left\{ 1,2,3Ê\right\}, \left\{ 2,3,4Ê\right\}, \left\{ 4\right\},\left\{ 4,5,6Ê\right\}, \left\{ 5,6,7Ê\right\} \right\}$ is not a grove. As above, the intuitive reason is that for partiton $\pi := \left\{ 1,2,3Ê\right\}, \left\{ 2,3,4Ê\right\}, \left\{ 4\right\} | \left\{ 4,5,6Ê\right\},$ \\$  \left\{ 5,6,7Ê\right\}$, no matter which topology assignment is chosen, the fact that $\sS_1$ and $\sS_2$ overlap in just one taxon, namely taxon 4, is not enough to fix the supertrees. Basically, one can choose compatible trees for all taxon sets such that supertrees for $\sS_1$ and $\sS_2$ are computed, respectively. Then, by Lemma 4.1 of \cite{ane_eulenstein_2009}, as the taxa of $\sS_1$ and $\sS_2$ only intersect in one taxon, the two trees are compatible and can be combined in any possible way to form a common supertree, as the only condition is that both contain taxon 4. Therefore, no cross triple of this partition is resolved (but there are cross triples with respect to $\pi$, e.g. $\{1,2,5\}$). Thus, $\sS_1 \cup \sS_2$ is not a grove.
\end{proof}

It is worth noting that the union of informative groves does not only lose informativeness when the underlying groves intersect, but rather the grove property itself. We now show a weaker statement for strictly informative groves.

\begin{prop} $\mbox{}$  \label{nonstrictinf}The union of intersecting strictly informative groves is not necessarily strictly informative (even if it is a grove). 
\end{prop}

\begin{proof} $\mbox{}$  We provide an explicit counterexample. Let $\sS_1:= \left\{ \left\{ 1,2,3Ê\right\}, \left\{ 2,3,4Ê\right\} \right\}$ and $\sS_2:= \left\{ \left\{ 1,2,3Ê\right\}, \left\{ 1,3,4Ê\right\} \right\}$. Then, $\sS_1$ is a strictly informative grove as the only way to partition $\sS_1$ is to split $\left\{ 1,2,3Ê\right\}$ and  $\left\{ 2,3,4Ê\right\}$ apart. The cross triples are $\left\{ 1,2,4Ê\right\}$ and $\left\{ 1,3,4Ê\right\}$, which are for instance resolved when the topology assignment $((1,2),3)$ and $((2,3),4)$ is chosen, as then the only possible supertree is $(((1,2),3),4)$. This is analogous to the situation depicted in Figure \ref{caterpillar}. By the same argument, $\sS_2$ is a strictly informative grove. Now $\sS_1 \cap \sS_2 =\left\{ \left\{ 1,2,3 \right\}\right\} \neq \emptyset$. But $\sS_1 \cup \sS_2 = \left\{ \left\{ 1,2,3Ê\right\}, \left\{ 2,3,4Ê\right\}, \left\{ 1,3,4\right\} \right\}$ is not strictly informative. This can be seen when examining the split $\left\{ 1,2,3Ê\right\}  | $\\$ \left\{ 2,3,4Ê\right\}, \left\{ 1,3,4\right\}$, which has no cross triples and therefore does not fulfill the strict informativeness criterion.
\end{proof}

While Proposition \ref{nonstrictinf} shows that even strictly informative groves do not fulfill the properties of Conjecture \ref{conjecture} (2), we conjecture that a scenario as the one in Proposition \ref{counterex2} is not possible for strictly informative groves.

\begin{conj} The union of two intersecting strictly informative groves is a grove. \end{conj}
\vspace{0.5cm}

\subsection{2-overlap groves}

As the previous sections show, the grove concept introduced by An\'e et al. \cite{ane_eulenstein_2009} does not fulfill the properties specified in Conjecture \ref{conjecture}, even if it is interpreted in a stricter way as in Section \ref{informativegroves}. As explained in Section \ref{introduction}, Conjecture \ref{conjecture} (1) is of importantance concerning the search for groves in databases.
We now introduce an alternative grove definition, which fulfills the conjecture and thus simpflifies the search for groves. In order to do so, we first need to define 2-overlap graphs.

\begin{defn} Let $\sS$ be a set of taxon sets $X_1,\ldots, X_k$. Its {\em 2-overlap graph is the graph} $G:=(\sS, E)$, where an edge $\{X_i,X_j\} \in E$ if and only if $X_i$ and $X_j$ share at least 2 taxa, i.e.  $|X_i \cap X_j| \geq 2$. 
\end{defn}

\begin{defn}[2-overlap grove] \label{2overlap}Let $\sS$ be a set of taxon sets. Then, $\sS$ is called a {\em 2-overlap grove} if and only if its 2-overlap graph is connected. \end{defn}

\begin{ex}Figure \ref{2overlapgraph} shows an example $\sS$ where the 2-overlap graph is connected which makes $\sS$ a 2-overlap grove, as well as an example $\hat{\sS}$, whose 2-overlap graph consists of two connected components such that $\hat{\sS}$ is not a 2-overlap grove.
\end{ex}

We now show that 2-overlap groves are indeed groves.

\begin{thm}[Theorem 1.1 of \cite{ane_eulenstein_2009}] Every set of taxon sets whose 2-overlap graph is connected is a grove.\label{2overgrove1}
\end{thm}

\begin{cor} Every 2-overlap grove is a grove.\label{2overgrove}
\end{cor}

The proof of Corollary \ref{2overgrove} is a direct conclusion of Theorem \ref{2overgrove1}. So all 2-overlap groves are groves, but -- as can be seen for instance in Proposition \ref{counterex} -- not all groves are 2-overlap groves. There, we used the groves $\sS_1:= \left\{ \left\{ 1,2,3Ê\right\}, \left\{ 1\right\} \right\}$ and $\sS_2:= \left\{ \left\{ 1,4,5Ê\right\}, \left\{ 1\right\} \right\}$, whose 2-overlap graphs are both not connected, to show that the union of intersecting groves need not be a grove. Next we show that such issues naturally cannot arise with 2-overlap groves.

\begin{thm} If two 2-overlap groves intersect, their union is also a 2-overlap grove.
\end{thm}

\begin{proof} Let $\sS_1$, $\sS_2$ be two 2-overlap groves such that $\sS_1 \cap \sS_2 \neq \emptyset$.Then, there is a taxon set $S \in \sS_1 \cap \sS_2$. Now, as the 2-overlap graphs of $\sS_1$ and $\sS_2$ are both connected, both are in particular connected to $S$. Thus, $S$ connects the two 2-overlap graphs and therefore the 2-overlap graph of $\sS:=\sS_1Ê\cup \sS_2$ is connected. By Definition \ref{2overlap}, this implies that $\sS$ is a 2-overlap grove. This completes the proof. \end{proof}

As Lemma \ref{equi} naturally also applies to 2-overlap groves, it turns out that 2-overlap groves have all required properties stated by Conjecture \ref{conjecture} (not just the second one), and thus resolve all problems inherent to the grove concept of Definition \ref{grove}. And as all 2-overlap groves are also groves, they also have the potential to construct informative supertrees. 

\begin{rem} Note that while all 2-overlap groves are also groves, they are not directly related to strictly informative groves, as neither property implies the other. We show this with the following two examples.
\begin{enumerate}
\item The set of taxon sets $\sS= \{\{1,2\},\{1,2,3\}\}$ is a 2-overlap grove, because its 2-overlap graph is connected as shown in Figure \ref{2overlapgraph}. However, as the only possible partition, namely the split $\{1,2\}|\{1,2,3\}$, does not have any cross triples and thus also no resolved ones, $\sS$ is a grove, but not a strictly informative one. So it has to be noted that the 2-overlap grove concept still includes some groves which contain one set with all taxa, in which case a supertree cannot provide new information.
\item The set of taxon sets $\hat{\sS}= \{\{2,3,6\},\{1,4,5\},\{2,4,5,7\}\}$ is a not a 2-overlap grove, because its 2-overlap graph is not connected as shown in Figure \ref{2overlapgraph}. However, choosing the topology assignment $\pP:=$\\$\{((2,3),6), ((1,4),5), (((4,5),2),7)\} $ leads to supertrees in which the position of the taxa $3$ and $6$ varies for each tree, but all contain the subtree $T:=((((1,4),5),2),7)$ (the so-called `maximum agreement subtree') as shown in Figure \ref{SIG}. Moreover, the triple $\{1,3,7\}$ is a cross triple for all possible partitions of $\hat{\sS}$, and this cross triple is resolved by the maximum agreement subtree. Thus, this topology assignment resolves a cross triple for all possible partitions. Therefore, $\hat{\sS}$ is a strictly informative grove.
\end{enumerate}
\end{rem}

  \begin{figure}[ht]      \centering\vspace{0.5cm} 
    \includegraphics[width=11cm]{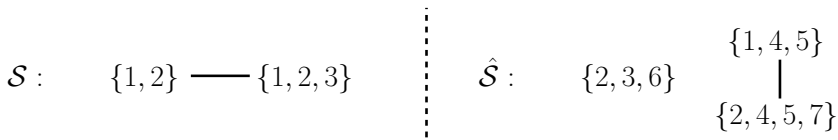} 
     \caption{The set $\sS= \{\{1,2\},\{1,2,3\}\}$ is a 2-overlap grove but not a strictly informative grove. For the set $\hat{\sS}$, the situation is vice versa.  }
   \label{2overlapgraph}
  \end{figure}
  
 \begin{figure}[ht]      \centering\vspace{0.5cm} 
    \includegraphics[width=11cm]{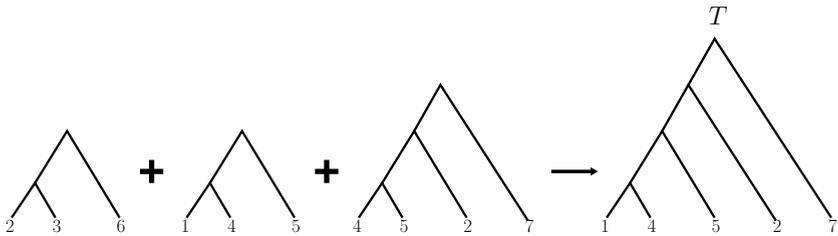} 
    \caption{The topology assignment $\pP:=\{((2,3),6), ((1,4),5), (((4,5),2),7)\} $ leads to various  supertrees, as the position of the taxa $3$ and $6$ varies. However, all supertrees contain the maximum agreement subtree $T:=((((1,4),5),2),7)$, so that in any case the triple $\{1,3,7\}$ gets resolved. }\label{SIG}
  \end{figure}

\section{Discussion}
In this paper, we showed that the concept of groves as introduced by An\'e et al. \cite{ane_eulenstein_2009} can be simplified to the much simpler concept of tripartition groves. This reduces the search for groves in databases drastically from investigating all partitions to analyzing only splits and tripartitions only, which is a great improvement. But we also showed that groves as introduced by An\'e et al. unfortunately do not fulfill the requirements of Conjecture \ref{conjecture}, which implies that they are hard to find in databases of taxon sets even when our simplification is considered. We also investigated slightly modified versions of the original grove definition in order to enforce more informativeness, but even then there were similar drawbacks in the grove concept.
Finally, we successfully proved that our new concept of 2-overlap groves fulfills all properties of the conjecture, while these 2-overlap groves are also groves in the original sense of the definition. So they comprise the good properties of groves, namely for instance offering the potential of informative supertrees, with a property that makes them more easily identifiable in databases. It should be noted that the latter is not only due to the fact that Conjecture \ref{conjecture} (2) holds for 2-overlap groves, but also to the fact that the construction of the 2-overlap graph of a database can be done more efficiently and more easily than investigating all cross triples of all possible splits and tripartitions. Moreover, the 2-overlap graph of a database immediately displays all maximal 2-overlap groves in the database and thus shows quickly which taxon sets should be combined in order to find new information in the corresponding phylogenetic supertrees.
It should be noted, though, that our concept of 2-overlap groves only covers a subset of the groves by An\'e et al. \cite{ane_eulenstein_2009}, as all 2-overlap groves are groves but the opposite does not hold. Therefore, some taxon sets which only overlap in a single taxon are not covered by our concept, whereas they may be considered as groves. However, we claim that these examples are rather artificial in the sense that in practice, biologists rarely wish to combine trees with only one taxon in common. So the simplification gained by restricting ourselves to 2-overlaps seems to outweigh the possible drawbacks by considering fewer cases. Altogether, we think this new concept is promising and more research should be done. For instance, practical analyses in real databases would be really helpful.


\subsection*{Acknowledgment}
I wish to thank Mike Steel, Arndt von Haeseler, Oliver Eulenstein and Anne Kupczok for helpful discussions. I also wish to thank Clemens Zarzer for his fruitful suggestions on previous versions of the manuscript as well as his programming support. Financial support from the Wiener Wissenschafts-, Forschungs- und Technologiefonds (WWTF) to Arndt von Haeseler is greatly appreciated.

\end{document}